\documentclass[copyright,creativecommons]{eptcs}
\providecommand{\event}{GandALF 2013} 
\usepackage{amsmath,amsthm,amssymb,bussproofs,epsfig,graphicx,verbatim}
\usepackage[roscoe]{csp-cm}
\usepackage[all]{xy}

\newcommand{\circled}[1]{\raisebox{.5pt}{\textcircled{\raisebox{-.9pt} {#1}}}}

\newcommand{\sqsub}{\,\raisebox{-.5ex}
  {$\stackrel{\textstyle\sqsubset}{\scriptstyle{\sim}}$}\,}

  \newcommand{\aasg}{\,\raisebox{0.065ex}{:}{=}\,}
\newcommand{\sbr}[1]{\lbrack \! \lbrack #1 \rbrack \! \rbrack}

\newcommand{\sq}{\textit{q}}
\newcommand{\sr}{\textit{run}}
\newcommand{\sok}{\textit{ok}}
\newcommand{\srd}{\textit{read}}
\newcommand{\sd}{\textit{done}}
\newcommand{\sw}{\textit{write}}

\newtheorem{proposition}{Proposition}
\newtheorem{theorem}{Theorem}
\newtheorem{definition}{Definition}

 \theoremstyle{definition}
   \newtheorem{example}{Example}

\title{Slot Games for Detecting Timing Leaks of Programs}
\author{Aleksandar S. Dimovski
\institute{Faculty of Information-Communication Tech., FON University, Skopje, 1000, MKD }
\email{aleksandar.dimovski@fon.edu.mk}
}

\def\titlerunning{Slot Games for Detecting Timing Leaks of Programs}
\def\authorrunning{A. S. Dimovski}
\begin{document}
\maketitle

\begin{abstract}
In this paper we describe a method for verifying secure information flow
of programs, where apart from direct and indirect flows
a secret information can be leaked through covert timing channels.
That is, no two computations of a program that differ only on
high-security inputs can be distinguished by low-security outputs and
timing differences. We attack this problem by using slot-game
semantics for a quantitative analysis of programs. We show how
slot-games model can be used for performing a precise security
analysis of programs, that takes into account both extensional and
intensional properties of programs. The practicality of this
approach for automated verification is also shown.
\end{abstract}

\section{Introduction}

Secure information flow analysis is a technique which performs a static analysis
of a program with the goal of proving that it will not leak any sensitive (secret)
information improperly. If the program passes the test, then we say that it is
secure and can be run safely. There are several ways in which secret information
can be leaked to an external observer. The most common are direct and indirect
leakages, which are described by the so-called non-interference property \cite{Gog,SS}.
We say that a program satisfies the non-interference property if its
high-security (secret) inputs do not affect its low-security (public) outputs,
which can be seen by external observers.

However, a program can also leak information through its timing behaviour,
where an external observer can measure its total running time.
Such timing leaks are difficult to detect and prevent, because they can exploit low-level
implementation details. To detect timing 
leaks, we need to ensure that 
the total running time of a program
do not depend on its high-security inputs.

In this paper we describe a game semantics based approach for performing
a precise security analysis.
We have already shown in \cite{D13} how game semantics can be applied for verifying the
non-interference property.
Now we use slot-game semantics to check for timing leaks of closed
and open programs. We focus here only on detecting covert timing channels, since
the non-interference property can be verified similarly as in \cite{D13}.
Slot-game semantics was developed in \cite{Gh05} for a quantitative analysis of
Algol-like programs.
It is suitable for verifying the above security properties, since it takes into
account both extensional (\emph{what} the program computes) and intensional
(\emph{how} the program computes) properties of programs.
It represents a kind of denotational semantics induced by the theory
of operational improvement of Sands \cite{Sands}. Improvement is a refinement
of the standard theory of operational approximation, where we say that one
program is an improvement of another if its execution is more efficient
in any program context.
We will measure efficiency of a program as the sum of costs associated with basic operations
it can perform.
It has been shown that slot-game semantics
is fully abstract (sound and complete) with respect to operational improvement,
so we can use it as a denotational theory of improvement to analyse
programming languages.

The advantages of game semantics (denotational) based approach for verifying security are several.
We can reason about open programs, i.e.\ programs with non-locally defined identifiers.
Moreover, game semantics is  compositional, which enables analysis about program
fragments to be combined into an analysis of a larger program.
Also the model hides the details of local-state manipulation of a program,
which results in small models with maximum level of abstraction where are
represented only visible input-output behaviours enriched with costs that
measure their efficiency.
All other behaviour is abstracted away, which makes this model very suitable
for security analysis.
Finally, the game model for some language fragments admits finitary representation
by using regular languages or CSP processes \cite{GM,DL}, and has already been
applied to automatic program verification. Here we present another application of
algorithmic game semantics
for automatically verifying security properties of programs.

\paragraph{Related work.}
The most common approach to ensure security properties of programs is by using
security-type systems \cite{Slam}.
Here for every program component are defined security types,
which contain information about their types and security levels.
Programs that are well-typed under these type systems satisfy certain security properties.
Type systems for enforcing non-interference of programs have been proposed by
Volpano and Smith in \cite{Vol}, and subsequently they have been extended to
detect also covert timing channels in \cite{Vol97,Agat}.
A drawback of this approach is its imprecision, since many secure programs are not typable
and so are rejected. A more precise analysis of programs can
be achieved by using semantics-based approaches \cite{Joshi}.

\section{Syntax and Operational Semantics} \label{lang}

We will define a secure information flow analysis for
Idealized Algol (IA), a small Algol-like language introduced
by Reynolds \cite{Rey1} which has been used as a metalanguage in
the denotational semantics community. It is a  call-by-name $\lambda$-calculus
extended with imperative features and locally-scoped variables.
In order to be able to perform an automata-theoretic analysis of the language,
we consider here its second-order recursion-free fragment (IA$_2$ for short).
It contains finitary data types $D$: $\mathsf{int}_n=\{ 0, \ldots, n-1 \}$ and $\mathsf{bool}=\{ tt, ff \}$,
and first-order function types: $T ::= B \mid B \rightarrow T$, where $B$
ranges over base types: expressions ($\mathsf{exp}D$), commands ($\mathsf{com}$),
and variables ($\mathsf{var}D$).

Syntax of the language is given by the following grammar:
\[
\begin{array}{l}
M ::= \!\! x \! \mid \! v \! \mid \! \mathsf{skip} \mid \! \mathsf{diverge} \mid M\, \mathsf{op} \, M \mid M ; \!\! M \mid
 \mathsf{if}\,M\,\mathsf{then}\,M \,\mathsf{else}\, M \! \mid \! \mathsf{while}\,M\,\mathsf{do}\,M   \\
 \ \qquad \mid  M := M \mid !M \mid \mathsf{new}_D \: x\!:=\!v \: \mathsf{in} \: M
  \mid \mathsf{mkvar}_D MM \mid \! \lambda x . M \mid MM
\end{array}
\]
where $v$ ranges over constants of type $D$.

Typing judgements are of the form $\Gamma
\vdash M : T$, where $\Gamma$ is a type \emph{context} consisting of
a finite number of typed free identifiers.
Typing rules of the language are standard \cite{AM2}, but the general
application rule is broken up into the linear application and the contraction rule
\footnote{ $M[N/x]$ denotes the capture-free substitution of $N$ for $x$ in $M$.}.
\begin{center}
\AxiomC{$\Gamma \vdash M : B \to T $} \AxiomC{$\Delta \vdash N : B $}
 \BinaryInfC{$\Gamma, \Delta \vdash MN : T$}
\DisplayProof \qquad
\AxiomC{$\Gamma, x_1:T, x_2:T \vdash M : T'$}
 \UnaryInfC{$\Gamma, x:T \vdash M[x/x_1, x/x_2]:T' $}
 \DisplayProof
\end{center}
We use these two rules to have control over multiple occurrences of
free identifiers in terms during typing.

Any input/output operation in a term is done through global variables, i.e.\ free
identifiers of type $\mathsf{var}D$.
So an input is read by de-referencing a global variable, while an output is written
by an assignment to a global variable.

The operational semantics is defined in terms of a small-step evaluation
relation using a notion of an evaluation context \cite{Fel}.
A small-step evaluation (reduction) relation is of the form:
\[
\Gamma \vdash M,\mathrm{s} \longrightarrow M',\mathrm{s}'
\]
where $\Gamma$ is a so-called $\mathsf{var}$-context which contains
only identifiers of type $\mathsf{var}D$; $\mathrm{s}$, $\mathrm{s}'$
are $\Gamma$-states which assign data values to the variables in $\Gamma$;
and $M$, $M'$ are terms.
The set of all $\Gamma$-states will be denoted by $St(\Gamma)$.

Evaluation contexts are contexts
\footnote{A context $C[-]$ is a term with (several occurrences of) a hole in it,
such that if $\Gamma\vdash M:T$ is a term of the same type as the hole
then $C[M]$ is a well-typed closed term of type
$\mathsf{com}$, i.e.\ $\vdash C[M]:\mathsf{com}$.}
containing a single hole which is
used to identify the next sub-term to be evaluated (reduced).
They are defined inductively by the following grammar:
\[
\begin{array}{l}
E ::=  [-]  \mid  E M  \mid E ; \! M  \mid \mathsf{skip} ; \! E  \mid  E \, \mathsf{op} \, M
\mid v \, \mathsf{op} \, E \mid  \mathsf{if}\, E \,\mathsf{then}\,M \,\mathsf{else}\, M  \mid  M := E  \mid    E := v  \mid !E
\end{array}
\]

The operational semantics is defined in two stages. First,
a set of basic reduction rules are defined in Table~\ref{red_rul}.
We assign different (non-negative) costs to each reduction rule,
in order to denote how much computational time is needed for
a reduction to complete.
They are only descriptions of time and we can give them different
interpretations describing how much real time they denote.
Such an interpretation can be arbitrarily complex. So the semantics
is parameterized on the interpretation of costs.
Notice that we write
$\mathrm{s} \otimes (x \mapsto v)$ to denote a $\{\Gamma,x\}$-state
which properly extends $\mathrm{s}$ by mapping $x$ to the value $v$.

\begin{table}
\fbox{
\begin{minipage}{87ex}
$ \begin{array}{@{}l}
 \Gamma \vdash n_1 \, \mathsf{op} \, n_2, \mathrm{s} \longrightarrow^{k_{op}} n, \mathrm{s}, \ \textrm{where} \, n=n_1 \mathrm{op} n_2 \\
 \Gamma \vdash \mathsf{skip} ; \mathsf{skip}, \mathrm{s} \longrightarrow^{k_{seq}} \mathsf{skip}, \mathrm{s}' \\
 \Gamma \vdash \mathsf{if} \, tt \, \mathsf{then} \, M_1 \, \mathsf{else} M_2, \mathrm{s} \longrightarrow^{k_{if}} M_1, \mathrm{s} \\
 \Gamma \vdash \mathsf{if} \, ff \, \mathsf{then} \, M_1 \, \mathsf{else} M_2, \mathrm{s} \longrightarrow^{k_{if}} M_2, \mathrm{s} \\
 \Gamma \vdash x \aasg v', \mathrm{s} \otimes (x \mapsto v) \longrightarrow^{k_{asg}} \mathsf{skip}, \mathrm{s} \otimes (x \mapsto v') \\
 \Gamma \vdash !x, \mathrm{s} \otimes (x \mapsto v) \longrightarrow^{k_{der}} v, \mathrm{s} \otimes (x \mapsto v) \\
 \Gamma \vdash (\lambda x.M)M', \mathrm{s} \longrightarrow^{k_{app}} M[M'/x], \mathrm{s} \\
  \Gamma \vdash \mathsf{new}_D \, x \aasg v \, \mathsf{in} \, \mathsf{skip}, \mathrm{s} \longrightarrow^{k_{new}} \mathsf{skip}, \mathrm{s}
 \end{array} $
\end{minipage}
} \caption{Basic Reduction Rules } \label{red_rul}
\end{table}

We also have reduction rules for iteration, local variables, and $\mathsf{mkvar}_D$ construct, which do not incur additional costs.
\begin{center}
$\Gamma \vdash \mathsf{while} \, b \, \mathsf{do} \, M, \mathrm{s} \longrightarrow \mathsf{if} \, b \, \mathsf{then} \, (M ; \mathsf{while} \, b \,
\mathsf{do} \, M) \, \mathsf{else \, skip}, \mathrm{s}$ \\
\AxiomC{$\Gamma, y \vdash M [y/x], \mathrm{s} \otimes (y \mapsto v) \longrightarrow M', \mathrm{s'} \otimes (y \mapsto v')$}
 \UnaryInfC{$\Gamma \vdash \mathsf{new}_D \, x \aasg v \, \mathsf{in} \, M, \mathrm{s} \longrightarrow
 \mathsf{new}_D \, x \aasg v' \, \mathsf{in} \, M'[x/y], \mathrm{s'}$}
 \DisplayProof \\
 $\Gamma \vdash (\mathsf{mkvar}_D \, M_1 M_2) \aasg v, \mathrm{s} \longrightarrow M_1 v, \mathrm{s} \qquad
 \Gamma \vdash !(\mathsf{mkvar}_D \, M_1 M_2), \mathrm{s} \longrightarrow M_2, \mathrm{s}$
\end{center}

Next, the in-context reduction rules for arbitrary terms are defined as:
\begin{center}
\AxiomC{$\Gamma \vdash M, \mathrm{s} \longrightarrow^{n} M', \mathrm{s'} $}
 \UnaryInfC{$\Gamma \vdash E[M], \mathrm{s} \longrightarrow^{n} E[M'], \mathrm{s'} $}
 \DisplayProof
\end{center}
The small-step evaluation relation is deterministic, since arbitrary term
can be uniquely partitioned into an evaluation context and a sub-term,
which is next to be reduced.

We define the reflexive and transitive closure of the small-step reduction
relation as follows:
\begin{center}
\AxiomC{$\Gamma \vdash M, \mathrm{s} \longrightarrow^{n} M', \mathrm{s'} $}
 \UnaryInfC{$\Gamma \vdash M, \mathrm{s} \rightsquigarrow^{n} M', \mathrm{s'}$}
\DisplayProof \qquad
\AxiomC{$\Gamma \vdash M, \mathrm{s} \rightsquigarrow^{n} M', \mathrm{s'} $}
\AxiomC{$\Gamma \vdash M', \mathrm{s'} \rightsquigarrow^{n'} M'', \mathrm{s''} $}
 \BinaryInfC{$\Gamma \vdash M, \mathrm{s} \rightsquigarrow^{n+n'} M'', \mathrm{s''}$}
 \DisplayProof
\end{center}

Now a theory of operational improvement is defined \cite{Sands}.
Let $\Gamma \vdash M : \mathsf{com}$ be a term, where $\Gamma$ is a $\mathsf{var}$-context.
We say that $M$ \emph{terminates in $n$ steps} at
state $\mathrm{s}$, written $M,\mathrm{s} \Downarrow^n$,
if $\Gamma \vdash M,\mathrm{s} \rightsquigarrow^{n}
\mathsf{skip},\mathrm{s}'$ for some state $\mathrm{s}'$.
If $M$ is a closed term and $M, \emptyset \Downarrow^n$, then
 we write $M \Downarrow^n$.
If $M \Downarrow^n$ and $n \leq n'$, we write $M \Downarrow^{\leq n'}$.
We say that a term $\Gamma\vdash M:T$ may be \emph{improved}
by $\Gamma\vdash N:T$, denoted by $\Gamma\vdash M \gtrsim N$, if and
only if for all contexts $C[-]$, if $C[M]\Downarrow^n$  then $C[N]\Downarrow^{\leq n}$.
If two terms improve each other they are considered
\emph{improvment-equivalent}, denoted by
$\Gamma\vdash M \thickapprox N$.


Let $\Gamma, \Delta \vdash M:T$ be a term where $\Gamma$
is a $var$-context and $\Delta$ is an arbitrary context.
Such terms are called \emph{split terms}, and we denote them as $\Gamma | \Delta \vdash M:T$.
If $\Delta$ is empty, then these terms are called \emph{semi-closed}.
The semi-closed terms have only some global variables, and
the operational semantics is defined only for them.
We say that  a semi-closed
term $h:\mathsf{varD} | - \vdash M:\mathsf{com}$ does not have \emph{timing leaks}
if the initial value of the high-security variable $h$ does not influence the number of reduction steps of $M$.
More formally, we have:
\begin{definition}
A semi-closed term $h:\mathsf{varD} | - \vdash M:\mathsf{com}$ has no \emph{timing leaks} if
\begin{equation} \label{for1}
\begin{array}{ll}
\forall s_1, s_2 \in St(\{h\}). & s_1(h) \neq s_2(h) \land \\
& h:\mathsf{varD} \vdash M,\mathrm{s_1} \rightsquigarrow^{n_1} \mathsf{skip},\mathrm{s_1}' \land
h:\mathsf{varD} \vdash M,\mathrm{s_2} \rightsquigarrow^{n_2} \mathsf{skip},\mathrm{s_2}' \\
& \implies n_1=n_2
\end{array}
\end{equation}
\end{definition}
\begin{definition} \label{int-split}
We say that a \emph{split term} $h:\mathsf{varD} | \Delta \vdash M:\mathsf{com}$
does not have timing leaks, where $\Delta=x_1 : T_1, \ldots, x_k :
T_k$, if for all closed terms $\vdash N_1 : T_1, \ldots, \vdash N_k : T_k$,
 we have that the term $h:\mathsf{varD} | - \vdash M[N_1/x_1, \ldots, N_k/x_k]:\mathsf{com}$
 does not have timing leaks.
 \end{definition}

The formula (\ref{for1}) can be replaced by an equivalent formula,
where instead of two evaluations of the same term  we can consider
only one evaluation of the sequential composition of the given
term with another its copy \cite{Bar}.
So sequential composition enables us to place these two evaluations one after the other.
Let $h:\mathsf{varD} \vdash M:\mathsf{com}$ be a term, we define
$M'$ to be $\alpha$-equivalent to $M[h'/h]$  where all bound variables are suitable renamed.
The following can be shown:
$h \vdash M,\mathrm{s_1} \rightsquigarrow^{n} \mathsf{skip},\mathrm{s_1}' \land
h' \vdash M',\mathrm{s_2} \rightsquigarrow^{n'} \mathsf{skip},\mathrm{s_2}'$ iff
$h, h' \vdash M ; M',\mathrm{s_1} \otimes \mathrm{s_2} \rightsquigarrow^{n+n'} \mathsf{skip; skip},\mathrm{s_1}' \otimes
\mathrm{s_2}'$.
In this way, we provide an alternative definition to formula (\ref{for1}) as follows.
We say that a semi-closed term $h | - \vdash M:T$ has no \emph{timing leaks} if
\begin{equation} \label{for2}
\begin{array}{ll}
\forall s_1 \in St(\{h\}), s_2 \in St(\{h'\}). & \ s_1(h) \neq s_2(h') \land \\
& h, h' \vdash M ; M',\mathrm{s_1} \otimes \mathrm{s_2} \rightsquigarrow^{n_1} \mathsf{skip} ; M',\mathrm{s_1}' \otimes \mathrm{s_2}
\rightsquigarrow^{n_2} \mathsf{skip; skip},\mathrm{s_1}' \otimes \mathrm{s_2}' \\
&\implies n_1=n_2
\end{array}
\end{equation}

\section{Algorithmic Slot-Game Semantics} \label{game}

We now show how slot-game semantics for IA$_2$ can be represented algorithmically
by regular-languages.
In this approach, types are interpreted as games, which have two participants:
the Player representing the term, and the Opponent
representing its context. A game (arena) is defined by means of a set of moves,
each being either a question move or an answer move.
Each move represents an observable action that a term of
a given type can perform.
Apart from moves, another kind of action, called \emph{token} (slot), is used
to take account of quantitative aspects of terms.
It represents a payment that a participant needs to pay in order to use a resource such as time.
A computation is interpreted as a play-with-costs, which is given as a
sequence of moves and token-actions played by two participants in turns.
We will work here with complete plays-with-costs which represent the
observable effects along with incurred costs of a completed computation.
Then a term is modelled by a strategy-with-costs, which is a set of complete
plays-with-costs.
In the regular-language representation of game semantics \cite{GM},
types (arenas) are expressed as \emph{alphabets of moves},
computations (plays-with-costs) as \emph{words}, and terms (strategies-with-costs) as \emph{regular-languages}
over alphabets.

Each type $T$ is interpreted by an alphabet of moves $\mathcal A_{\sbr{T}}$,
which can be partitioned into two subsets of \emph{questions} $Q_{\sbr{T}}$ and \emph{answers} $A_{\sbr{T}}$.
For expressions, we have: $Q_{\sbr{\mathsf{exp}D}} = \{ \sq \}$ and $A_{\sbr{\mathsf{exp}D}} = D$,
i.e.\ there are  a question move $\sq$ to ask for the value of the expression and
values from $D$ are possible answers.
For commands, we have: $Q_{\sbr{\mathsf{com}}} = \{ \sr \}$ and $A_{\sbr{\mathsf{com}}} = \{ \sd \}$,
i.e.\ there are a question move $\sr$ to initiate a command and an answer move $\sd$ to signal
successful termination of a command.
For variables, we have: $Q_{\sbr{\mathsf{var}D}} = \{ \srd, \sw(a) \, \mid \, a
\in D \}$ and $A_{\sbr{\mathsf{var}D}} = D \cup \{ \sok \}$, i.e.\ there are moves for writing to the variable, $\sw(a)$,
acknowledged by the move $\sok$, and for reading from the variable, we have a question move $\srd$,
and an answer to it can be any value from $D$.
For function types, we have $ \mathcal{A}_{\sbr{B_1^{1} \to \ldots \to B_k^{k} \to B}} = \sum_{1 \leq i \leq k} \mathcal{A}_{\sbr{B_i}}^{i} + \mathcal{A}_{\sbr{B}}$, where $+$ means a disjoint
union of alphabets. We will use superscript tags to keep
record from which type of the disjoint union each move comes from.
We denote the token-action by $\circled{\$}$. A sequence of $n$ token-actions $\circled{\$}$
will be written as $\circled{n}$.

\newcommand\trnc[2]{#1\upharpoonright #2}

For any ($\beta$-normal) term we define a
regular language specified by an \emph{extended regular expression} $R$.
Apart from the standard operations for generating regular expressions, we will use
some more specific operations.
We define composition of regular expressions $R$ defined over alphabet $\mathcal{A}^{1} + \mathcal{B}^{2} + \{ \circled{\$} \}$
and $S$ over $\mathcal{B}^{2} + \mathcal{C}^{3} + \{ \circled{\$} \}$ as follows:
\begin{equation*}
\begin{array}{l}
R \comp_{\mathcal{B}^{2}} S = \{ w \big[ s /  a^{2} \cdot b^{2} \big] \ \mid  \ w \in S,
 a^{2} \cdot s \cdot b^{2} \in R \}
\end{array}
\end{equation*}
where $R$ is a set of words of the form $a^{2} \cdot s \cdot b^{2}$,
such that $a^{2}$, $b^{2} \in \mathcal{B}^{2}$ and $s$ contains only letters
from $\mathcal A^{1}$ and $\{ \circled{\$} \}$. 
Notice that the composition is defined over $\mathcal{A}^{1} + \mathcal{C}^{3} + \{ \circled{\$} \}$,
and all letters of $\mathcal{B}^{2}$ are hidden.
The shuffle operation  $R \bowtie S$
generates the set of all possible interleavings from words of $R$ and $S$, and
the restriction operation $R \mid_{\mathcal A'}$ ($R$ defined over $\mathcal A$ and $\mathcal A' \subseteq \mathcal A$)
removes from words of $R$ all letters from $\mathcal A'$.

If $w$, $w'$ are words, $m$ is a move, and $R$ is a regular expression,
 define $m \cdot w \smallfrown w' = m \cdot w' \cdot w$,
and $R \smallfrown w' = \{ w \smallfrown w' \mid w \in R \}$.
Given a word with costs $w$ defined over $\mathcal{A}+\{ \circled{\$} \}$, we define
the underlying word of $w$ as $w^{\dag}=w \mid_{\{\circled{\$}\}}$,
and the cost of $w$ as $w \mid_{\mathcal{A}}=\circled{n}$, which we denote as $|w|=n$.

The regular expression for $\Gamma \vdash M : T$ is denoted $\sbr{\Gamma \vdash M : T}$
and is defined over the alphabet $\mathcal A_{\sbr{\Gamma \vdash T}} = \big( \sum_{x : T' \in
\Gamma} \mathcal{A}_{\sbr{T'}}^{x} \big) + \mathcal{A}_{\sbr{T}} + \{ \circled{\$} \}$.
Every word in $\sbr{\Gamma \vdash M : T}$ corresponds to a complete play-with-costs in the strategy-with-costs
for $\Gamma \vdash M : T$.

Free identifiers $x \in \Gamma$ are interpreted by the copy-cat regular expressions,
which contain all possible computations that terms of that type can have.
Thus they provide the most general closure of an open term.
\[
\begin{array}{@{}l}
 \sbr{\Gamma, x:B_1^{x,1} \to \ldots B_k^{x,k} \to B^{x} \vdash x:B_1^{1} \to \ldots B_k^{k} \to B} = \qquad \qquad \\
 \qquad  \qquad \qquad \qquad \displaystyle{\sum_{q \in Q_{\sbr{B}}}} q \cdot q^{x} \cdot \big( \displaystyle{\sum_{1 \leq i \leq k}} ( \sum_{q_1 \in Q_{\sbr{B_i}}} q_{1}^{x,i} \cdot q_{1}^{i} \cdot \sum_{a_1 \in A_{\sbr{B_i}}} a_{1}^{i} \cdot a_{1}^{x,i} ) \big)^* \cdot \sum_{a \in A_{\sbr{B}}} a^{x} \cdot a
\end{array}
\]
When a first-order non-local function is called, it may evaluate any of
its arguments, zero or more times, and then it can return any value from
its result type as an answer.
For example, the term $\sbr{\Gamma, x:\mathsf{exp}D^{x} \vdash x:\mathsf{exp}D}$ is modelled by
the regular expression: $\sq \cdot \sq^x \cdot \sum_{n \in D} n^x \cdot n$.

The linear application is defined as:
\[
\sbr{\Gamma, \Delta \vdash M \, N : T} = \sbr{\Delta \vdash N : B^1} \comp_{\mathcal A_{\sbr{B}}^1} \sbr{\Gamma \vdash M : B^1 \to T}
\]
Since we work with terms in $\beta$-normal form, function application can occur
only when the function term is a free identifier. In this case, the interpretation
is the same as above except that we add the cost $k_{app}$ corresponding to function application.
Notice that $k_{app}$ denotes certain number of $\circled{\$}$ units that
are needed for a function application to take place.
The contraction $\sbr{\Gamma, x:T^{x} \vdash M[x/x_1,x/x_2] : T'}$ is obtained from
$\sbr{\Gamma, x_1:T^{x_1}, x_2:T^{x_2} \vdash M : T'}$, such that the moves associated
with $x_1$ and $x_2$ are de-tagged so that they represent actions associated with $x$.

To represent local variables, we first need to define
a (storage) `cell' regular expression $\mathsf{cell}_v$ which imposes the good variable behaviour on the local
variable. So $\mathsf{cell}_v$ responds to each $\sw(n)$ with $\sok$, and plays
the most recently written value in response to $\srd$, or if no value
has been written yet then answers the $\srd$ with the initial value $v$. Then we have:
\[
\begin{array}{l}
\mathsf{cell}_v = (\srd \cdot v)^* \cdot \big( \displaystyle{\sum_{n \in D}} \sw(n) \cdot \sok \cdot (\srd \cdot n)^*   \big)^* \\
\sbr{\Gamma,x:\mathsf{var}D \vdash M} \circ  \mathsf{cell}_{v}^{x} = \big( \sbr{\Gamma,x:\mathsf{var}D \vdash M} \cap
( \mathsf{cell}_{v}^{x} \bowtie (\mathcal A_{\sbr{\Gamma \vdash B}}+\circled{\$})^* ) \big) \mid_{\mathcal A_{\sbr{\mathsf{var}D}}^x} \\
\sbr{\Gamma \vdash \mathsf{new}_D \, x \aasg v \, \mathsf{in} \, M} = \sbr{\Gamma,x:\mathsf{var}D \vdash M} \circ  \mathsf{cell}_{v}^{x} \smallfrown k_{var}
\end{array}
\]
Note that all actions associated with $x$ are hidden away in the model of $\mathsf{new}$,
since $x$ is a local variable and so not visible outside of the term.

Language constants and constructs are interpreted as follows:
\[ \begin{array}{l}
\sbr{\mathsf{v:exp}D}=\{\sq \cdot v \} \quad
\sbr{\mathsf{skip:com}}=\{ \sr \cdot \sd \} \ \
\sbr{\mathsf{diverge:com}}\!=\! \emptyset \\
\sbr{\mathsf{op:exp}D^1 \times \mathsf{exp}D^2 \to \mathsf{exp}D'}
= \sq \cdot k_{op} \cdot \sq^1 \cdot \sum_{m \in D} m^1 \cdot q^2 \! \cdot \! \sum_{n \in D} n^2 \! \cdot \! (m \, op \, n)   \\
\sbr{\mathsf{;:com}^1 \to \mathsf{com}^2 \to \mathsf{com}}
= \sr \cdot \sr^1 \cdot \sd^1 \cdot k_{seq} \cdot \sr^2 \cdot \sd^2 \cdot \sd \\

\sbr{\mathsf{if:expbool}^1 \to \mathsf{com}^2 \to \mathsf{com}^3 \to \mathsf{com}}
= \sr \cdot k_{if} \cdot \sq^1 \cdot tt^1 \cdot \sr^2 \cdot \sd^2 \cdot \sd \ + \\
\qquad \qquad \qquad \qquad \qquad \qquad \qquad \qquad \quad \sr \cdot k_{if} \cdot \sq^1 \cdot ff^1 \cdot \sr^3 \cdot \sd^3 \cdot \sd \\
\sbr{\mathsf{while:expbool}^1 \to \mathsf{com}^2 \to \mathsf{com}}
= \sr \cdot ( k_{if} \cdot \sq^1 \cdot tt^1 \cdot \sr^2 \cdot \sd^2 )^* \cdot k_{if} \cdot \sq^1 \cdot ff^1 \cdot \sd \\
\sbr{\mathsf{:=:var}D^1 \to \mathsf{exp}D^2 \to \mathsf{com}}
= \sum_{n \in D} \sr \cdot k_{asg} \cdot \sq^2 \cdot n^2 \cdot \sw(n)^1 \cdot \sok^1 \cdot \sd \\
\sbr{\mathsf{!:var}D^1 \to \mathsf{exp}D}
= \sum_{n \in D} \sq \cdot k_{der}  \cdot  \srd^1 \cdot n^1 \cdot n
\end{array}
\]
Although it is not important at what position in a word costs are placed,
for simplicity we decide to attach them just after the initial move.
The only exception is the rule for sequential composition ($;$), where
the cost is placed between two arguments. The reason will be explained later on.

We now show how slot-games model relates to the operational semantics.
First, we need to show how to represent the state explicitly in the model.
A $\Gamma$-state $\mathrm{s}$ is interpreted as follows:
\[
\sbr{\mathrm{s} : \mathsf{var}D_1^{x_1} \times \ldots \times \mathsf{var}D_k^{x_k}} =
 \mathsf{cell}_{\mathrm{s}(x_1)}^{x_1} \bowtie  \ldots \bowtie \mathsf{cell}_{\mathrm{s}(x_k)}^{x_k}
\]
The regular expression $\sbr{s}$ is defined over the alphabet $\mathcal A_{\sbr{\mathsf{var}D_1}}^{x_1}+\ldots+\mathcal A_{\sbr{\mathsf{var}D_k}}^{x_k}$,
and words in $\sbr{s}$ are such that projections onto $x_i$-component are the same
as those of suitable initialized $\mathsf{cell}_{\mathrm{s}(x_i)}$ strategies.
Note that $\sbr{s}$ is a regular expression without costs.
The interpretation of $\Gamma \vdash M : \mathsf{com}$ at state $\mathrm{s}$
is:
\[ \sbr{\Gamma\vdash M} \circ \sbr{\mathrm{s}} = \big( \sbr{\Gamma \vdash M} \cap
( \sbr{\mathrm{s}} \bowtie (\mathcal A_{\sbr{\mathsf{com}}}+\circled{\$})^* ) \big) \mid_{\mathcal A_{\sbr{\Gamma}}}
\]
which is defined over the alphabet $\mathcal A_{\sbr{\mathsf{com}}} + \{\circled{\$}\}$.
The interpretation $\sbr{\Gamma\vdash M} \circ \sbr{s}$ can be studied more closely by considering words
in which moves from $\mathcal A_{\sbr{\Gamma}}$ are not hidden.
Such words are called \emph{interaction sequences}.
For any interaction sequence $\sr \cdot t \cdot \sd \bowtie \circled{n}$ from $\sbr{\Gamma\vdash M} \circ \sbr{\mathrm{s}}$,
where $t$ is an even-length word over $\mathcal A_{\sbr{\Gamma}}$, we say that
it leaves the state $\mathrm{s}'$ if the last write moves in each
$x_i$-component are such that $x_i$ is set to the value $\mathrm{s}'(x_i)$.
For example, let $\mathrm{s}=(x \mapsto 1, y \mapsto 2)$, then the following interaction:
$\sr \cdot \sw(5)^y \cdot \sok^y \cdot \srd^x \cdot 1^x \cdot \sd$
leaves the state $\mathrm{s}'=(x \mapsto 1, y \mapsto 5)$.
Any two-move word of the form: $\sr^{x_i} \cdot n^{x_i}$ or  $\sw(n)^{x_i} \cdot \sok^{x_i}$
will be referred to as \emph{atomic state operation} of $\mathcal A_{\sbr{\Gamma}}$.
The following results are proved in \cite{Gh05} for the full ICA (IA plus parallel composition and semaphores),
but they also hold for the restricted fragment of it.

\begin{proposition}
If $\Gamma\vdash M:\{\mathsf{com,expD}\}$ and $\Gamma\vdash M,s \longrightarrow^n M',s'$, then
for each interaction sequence $i \cdot t$ from $\sbr{\Gamma\vdash M'} \circ \sbr{s'}$
($i$ is an initial move) there exists an interaction $i \cdot t_a \cdot t \smallfrown \circled{n} \in \sbr{\Gamma\vdash M} \circ \sbr{s}$
such that $t_a$ is an empty word or an atomic state operation of $\mathcal A_{\sbr{\Gamma}}$
which leaves the state $s'$.
\end{proposition}

\begin{proposition}
If $\Gamma \vdash M, \mathrm{s} \rightsquigarrow^{n} M', \mathrm{s'}$ then
$\sbr{\Gamma\vdash M'} \circ \sbr{s'} \bowtie \circled{n} \subseteq \sbr{\Gamma\vdash M} \circ \sbr{s}$.
\end{proposition}

\begin{theorem}[Consistency] \label{cons}
  If $M,\mathrm{s} \Downarrow^n$ then
  $\exists w \in \sbr{\Gamma\vdash M} \circ \sbr{s}$ such that $|w|=n$ and $w^{\dag}=\sr \cdot \sd$ .
\end{theorem}

\begin{theorem}[Computational Adequacy] \label{adeq}
  If $\exists w \in \sbr{\Gamma\vdash M} \circ \sbr{s}$ such that $|w|=n$ and $w^{\dag}=\sr \cdot \sd$, then
  $M,\mathrm{s} \Downarrow^n$.
\end{theorem}

We say that a regular expression $R$ is improved by $S$, denoted as $R \gtrsim S$,
if $\forall w \in R, \exists t \in S$,
  such that $w^{\dagger}=t^{\dagger}$ and $|w| \geq |t|$.

\begin{theorem}[Full Abstraction] \label{eq}
  $\Gamma\vdash M \gtrsim N$ iff $\sbr{\Gamma\vdash M} \gtrsim \sbr{\Gamma\vdash N}$.
\end{theorem}
This shows that the two theories of improvement based on operational and
game semantics are identical.

\section{Detecting Timing Leaks} \label{time}

In this section slot-game semantics is used to detect whether a term
with a secret global variable $h$ can leak information about the initial
value of $h$ through its timing behaviour.

For this purpose, we define a special command $\mathsf{skip}^{\#}$ which
similarly as $\mathsf{skip}$ does nothing, but its slot-game semantics
is: $\sbr{\mathsf{skip}^{\#}} = \{ \sr \cdot \# \cdot \sd \}$,
where $\#$ is a new special action, called
\emph{delimiter}. Since we verify security of a term by running
two copies of the same term one after the other, we will use the command
$\mathsf{skip}^{\#}$ to specify the boundary between these two copies.
In this way, we will be able to calculate running times of the two
terms separately.

\begin{theorem} \label{th.closed}
Let $h:\mathsf{varD} | - \vdash M:\mathsf{com}$ be a semi-closed term, and
\footnote{The free identifier $k$ in (\ref{for3}) is used to initialize the
variables $h$ and $h'$ to arbitrary values from $D$. }
\begin{equation} \label{for3}
\begin{array}{l}
R = \lbrack \! \lbrack k: \mathsf{exp}D \vdash \mathsf{new}_D \, h \aasg k \, \mathsf{in} \, M ; \mathsf{skip}^{\#} ;
 \mathsf{new}_D \, h' \aasg k \, \mathsf{in} \, M'  : \mathsf{com} \rbrack \! \rbrack
\end{array}
\end{equation}
Any word of $R$ is of the form $w=w_1 \cdot \# \cdot w_2$
such that $|w_1|=|w_2|$ iff $M$ has no timing leaks, i.e. the fact (\ref{for2}) holds.
\end{theorem}
\begin{proof}
Suppose that any word $w \in R$ is of the form $w=w_1 \cdot \# \cdot w_2$
such that $|w_1|=|w_2|$. Let us analyse the regular expression $R$ defined in (\ref{for3}).
We have:
\[
\begin{array}{l}
R = \{ \sr \cdot k_{var} \cdot \sq^{k} \cdot v^{k} \cdot w_1 \cdot k_{seq} \cdot \# \cdot k_{seq} \cdot k_{var} \cdot  \sq^{k} \cdot v'^{k} \cdot w_2 \cdot \sd \mid  \\
\qquad \qquad \sr \cdot w_1 \cdot \sd \in \sbr{ h \vdash M } \circ \mathsf{cell}_{v}^{h}, \sr \cdot w_2 \cdot \sd \in \sbr{ h' \vdash M' } \circ \mathsf{cell}_{v'}^{h'} \}
\end{array}
\]
for arbitrary values $v, v' \in D$.
In order to ensure that one $k_{seq}$ unit of cost occurs before and after the delimiter action,
$k_{seq}$ is played between
two arguments of the sequential composition as was described in Section \ref{game}.
Given that $\sr \cdot w_1 \cdot \sd \in \sbr{ h \vdash M } \circ \mathsf{cell}_{v}^{h}$ and
$\sr \cdot w_2 \cdot \sd \in \sbr{ h' \vdash M' } \circ \mathsf{cell}_{v'}^{h'}$ for any $v,v' \in D$,
by Computational Adequacy we have that $M,(h \mapsto v) \Downarrow^{|w_1|}$ and $M',(h' \mapsto v') \Downarrow^{|w_2|}$.
Since $|w_1|=|w_2|$, it follows that the fact (\ref{for2}) holds.

Let us consider the opposite direction. Suppose that the fact (\ref{for2}) holds.
The term in (\ref{for3}) is $\alpha$-equivalent to
$k \vdash \mathsf{new}_D \, h \aasg k \, \mathsf{in} \, \mathsf{new}_D \, h' \aasg k \, \mathsf{in} \, M ; \mathsf{skip}^{\#} ; M'$.
Consider $\sbr{h,h' \vdash M ; \mathsf{skip}^{\#} ; M'} \circ \sbr{(h \mapsto v) \otimes (h' \mapsto v')}$,
where $v,v' \in D$.
By Consistency, we have that $\exists w_1 \in \sbr{h,h' \vdash M} \circ \sbr{(h \mapsto v) \otimes (h' \mapsto v')}$
such that $|w_1|=n$ and $w_1$ leaves the state $(h \mapsto v_1) \otimes (h' \mapsto v')$,
and $\exists w_2 \in \sbr{h,h' \vdash M'} \circ \sbr{(h \mapsto v_1) \otimes (h' \mapsto v')}$
such that $|w_2|=n$ and $w_2$ leaves the state $(h \mapsto v_1) \otimes (h' \mapsto v'_1)$.
Any word $w \in R$ is obtained from $w_1$ and $w_2$ as above ($|w_1|=|w_2|$), and so satisfies
the requirements of the theorem.
\end{proof}

We can detect timing leaks from a semi-closed term by verifying that all words
in the model in (\ref{for3}) are in the required form. To do this, we restrict
our attention only to the costs of words in $R$.

\begin{example}
Consider the term:
\[
h : \mathsf{var \, int_2}  \vdash \mathsf{if} \, (!h>0) \, \mathsf{then } \, h \aasg !h+1; \, \mathsf{else} \, \mathsf{skip : com}
\]
The slot-game semantics of this term extended as in ($\ref{for3}$) is:
\[
\begin{array}{l}
\sr \cdot k_{var} \cdot \sq^k \cdot \big( 0^k \cdot k_{seq} \cdot \# \cdot k_{seq} \cdot k_{var} \cdot \sq^k \cdot (0^k \cdot \sd + 1^k  \cdot k_{der} \cdot k_{+} \cdot \sd) \\
\qquad \qquad \qquad + 1^k \cdot k_{seq} \cdot k_{der} \cdot k_{+} \cdot \# \cdot k_{seq} \cdot k_{var} \cdot \sq^k \cdot (0^k \cdot \sd + 1^k \cdot k_{der} \cdot k_{+} \cdot \sd) \big)
\end{array}
\]
This model includes all possible observable interactions
of the term with its environment, which contains only the identifier
$k$, along with the costs measuring its running time. Note that
the first value for $k$ read from the environment is used to initialize $h$,
while the second value for $k$ is used to initialize $h'$.

By inspecting we can see that the model contains the word:
\[
\sr \cdot k_{var} \cdot \sq^k \cdot 0^k \cdot k_{seq} \cdot \# \cdot k_{seq} \cdot k_{var} \cdot \sq^k \cdot 1^k \cdot  k_{der} \cdot k_{+} \cdot \sd
\]
which is not of the required form.
This word (play) corresponds to two computations of the given term where
initial values of $h$ are 0 and 1 respectively, such that the cost of the
second computation has additional $k_{der}+k_{+}$ units more than the first one.
\qed
\end{example}

We now show how to detect timing leaks of a split (open) term $h:\mathsf{varD} | \Delta \vdash M:\mathsf{com}$,
where $\Delta=x_1:T_1, \ldots, x_k:T_k$. To do this, we need to check timing efficiency
of the following model:
\begin{equation} \label{for4}
\sbr{h, h':\mathsf{varD} \vdash M[N_1/x_1, \ldots, N_k/x_k] ; \mathsf{skip}^{\#} ; M'[N_1/x_1, \ldots, N_k/x_k]}
\end{equation}
at state $(h \mapsto v, h' \mapsto v')$, for any closed terms $\vdash N_1:T_1, \ldots, \vdash N_k:T_k$,
and for any values $v,v' \in D$.
As we have shown slot-game semantics respects  theory of operational improvement, so we will need to
examine whether all its complete plays-with-costs $s$ are of the form $s_1 \cdot \# \cdot s_2$
where $|s_1|=|s_2|$. However, the model in (\ref{for4}) can not be represented as
a regular language, so it can not be used directly for detecting timing leaks.

Let us consider more closely the slot-game model in (\ref{for4}).
Terms $M$ and $M'$ are run in the same context $\Delta$, which
means that each occurrence of a free identifier $x_i$ from $\Delta$ behaves
uniformly in both $M$ and $M'$.
So any complete play-with-costs of the model in (\ref{for4}) will be a concatenation of complete
plays-with-costs from models for $M$ and $M'$ with additional constraints that behaviours
of free identifiers from
$\Delta$ are the same in $M$ and $M'$.
If these additional constraints are removed from the above model,
then we generate a model which is an over-approximation of it and where
free identifiers from $\Delta$ can behave freely in
$M$ and $M'$. Thus we obtain:
\[
\begin{array}{l}
\sbr{h, h':\mathsf{varD}  \vdash M[N_1/x_1, \ldots, N_k/x_k] ; \! \mathsf{skip}^{\#} ; \! M'[N_1/x_1, \ldots, N_k/x_k]}
\subseteq \qquad \\
\qquad \qquad \qquad \qquad \qquad \qquad \sbr{h, h':\mathsf{varD}  \vdash M ; \! \mathsf{skip}^{\#} ; \! M'[N_1/x_1, \ldots, N_k/x_k] }
\end{array}
\]
If $\vdash N_1:T_1, \ldots, \vdash N_k:T_k$ are arbitrary closed terms,
then they are interpreted by identity (copy-cat) strategies corresponding to their types, and so we have:
\[
\sbr{h, h':\mathsf{varD} \vdash M ; \! \mathsf{skip}^{\#} ; \! M'[N_1/x_1, \ldots, N_k/x_k] } = \sbr{ h, h':\mathsf{varD}, \Delta \vdash M ; \! \mathsf{skip}^{\#} ; \! M' }
\]
This model is a regular language and we can use it to detect timing leaks.
\begin{theorem}
Let $h:\mathsf{varD} | \Delta \vdash M:\mathsf{com}$ be a split (open) term, where
 $\Delta=x_1:T_1, \ldots, x_k:T_k$, and
\begin{equation} \label{for5}
\begin{array}{l}
S = \lbrack \! \lbrack k: \mathsf{exp}D, \Delta \vdash \mathsf{new}_D \, h \aasg k \, \mathsf{in} \, M ; \mathsf{skip}^{\#} ;
 \mathsf{new}_D \, h' \aasg k \, \mathsf{in} \, M' : \mathsf{com} \rbrack \! \rbrack
\end{array}
\end{equation}
If any word of $S$ is of the form $w=w_1 \cdot \# \cdot w_2$ such that $|w_1|=|w_2|$, Then
$h:\mathsf{varD} | \Delta \vdash M$ has no timing leaks.
\end{theorem}
Note that the opposite direction in the above result does not hold.
That is, if there exists a word from $S$ which is not of the required form
then it does not follow that $M$ has timing leaks, since the found word (play)
may be spurious introduced due to over-approximation in the model in (\ref{for5}),
and so it may be not present in the model in (\ref{for4}).

\begin{example}
Consider the term:
\[
h : \mathsf{var int_2}, f : \mathsf{expint_2}^{f,1} \to \mathsf{com}^{f}  \vdash f( !h) : \mathsf{com}
\]
where $f$ is a non-local call-by-name function.

The slot-game model for this term is as follows:
\[
\sr \cdot  k_{app} \cdot \sr^{f} \cdot  ( \sq^{f,1} \cdot k_{der} \cdot \srd^{h} \cdot  ( 0^{h} \cdot 0^{f,1} + 1^{h} \cdot 1^{f,1} ) )^* \cdot \sd^{f} \cdot \sd
\]
Once $f$ is called, it may evaluate its argument, zero or more times, and
then it  terminates successfully.
Notice that moves tagged with $f$ represent the actions of calling and
returning from the function $f$, while moves tagged with $f,1$ indicate actions
of the first argument of $f$.

If we generate the slot-game model of this term extended as in (\ref{for5}),
we obtain a word which is not in the required form:
\[
\begin{array}{l}
\sr \cdot k_{var} \cdot \sq^k \cdot 0^k  \cdot k_{app} \cdot \sr^f \cdot  \sq^{f,1} \cdot k_{der} \cdot 0^{f,1} \cdot \sd^{f} \cdot k_{seq} \cdot \# \cdot  k_{seq} \cdot k_{var} \cdot \sq^k \cdot 1^k \cdot  k_{app} \cdot \sr^{f} \cdot  \sd^{f} \cdot  \sd
\end{array}
\]
This word corresponds to two computations of the term, where the first one calls $f$ which evaluates its argument once,
and the second calls $f$ which does not evaluate its argument at all.
The first computation will have the cost of $k_{der}$ units more that the second one.
However, this is a spurious counter-example,
since $f$ does not behave uniformly in the two computations, i.e.\ it calls
its argument in the first but not in the second computation.
\qed
\end{example}

To handle this problem, we can generate an under-approximation of the model given in (\ref{for4}) which
can be represented as  a regular language.
Let $h:\mathsf{varD} \mid \Delta \vdash M$ be a term derived without using the contraction rule
for any identifier from $\Delta$. Consider the following model:
\begin{equation} \label{for6}
\begin{array}{l}
\sbr{ h,h':\mathsf{varD} \mid \Delta \vdash M ; \mathsf{skip}^{\#} ; M'}^m = \sbr{ h,h':\mathsf{varD} \mid \Delta \vdash M ; \mathsf{skip}^{\#} ; M'} \ \cap \qquad \qquad \\
\qquad \qquad \qquad \qquad \qquad \qquad \qquad ( \mathsf{delta}_{T_1,m}^{x_1} \bowtie \ldots  \bowtie \mathsf{delta}_{T_k,m}^{x_k} \bowtie (\mathcal A_{\sbr{h,h':\mathsf{varD} \vdash \mathsf{com}}}\!+\!\!\circled{\$})^* )
\end{array}
\end{equation}
where $m \geq 0$ denotes the number of times that free identifiers of function types
may evaluate its arguments at most.
The regular expressions $\mathsf{delta}_{T,m}$ are used to repeat zero or once an arbitrary behaviour
for terms of type $T$, and are defined as follows.
\[ \begin{array}{l}
\mathsf{delta}_{\mathsf{exp}D,0} = \sq \cdot \sum_{n \in D} n \cdot (\epsilon + \sq \cdot n) \qquad
\mathsf{delta}_{\mathsf{com},0} = \sr \cdot \sd \cdot (\epsilon + \sr \cdot \sd) \\
\mathsf{delta}_{\mathsf{var}D,0} = ( \srd \cdot \sum_{n \in D} n \cdot (\epsilon + \srd \cdot n )) \ + \
  ( \sum_{n \in D} \sw(n) \cdot \sok \cdot (\epsilon + \sw(n) \cdot \sok ) )
\end{array}
\]
If $T$ is a first-order function type, then $\mathsf{delta}_{T,m}$ will be
a regular language only when the number of times its arguments can be evaluated is limited.
For example, we have that:
\[
\mathsf{delta}_{\mathsf{com}^1 \to \mathsf{com},m} = \sr \cdot  \sum_{r=0}^{m} ( \sr^1 \cdot \sd^1 )^r \cdot \sd \cdot
(\epsilon + \sr \cdot ( \sr^1 \cdot \sd^1 )^r \cdot  \sd )
\]
If $T$ is a function type with $k$ arguments, then we have to remember
not only how many times arguments are evaluated in the first call, but also the exact
order in which arguments are evaluated.

Notice that we allow an arbitrary behavior of type $T$
to be repeated zero or once in $\mathsf{delta}_{T,m}$, since it is possible
that depending on the current value of $h$
an occurrence of a free identifier from
$\Delta$ to be run in $M$ but not in $M'$, or vice versa. For example, consider the term:
\[
h:\mathsf{var \, int}_2 \mid x,y : \mathsf{exp \, int}_2  \vdash \mathsf{new}_{int_2} \, z \aasg 0 \, \mathsf{in \, if} \, ( !h > 0 ) \, \mathsf{then} \ z \aasg x \, \mathsf{else} \ z \aasg y+1
\]
This term has timing leaks, and the corresponding counter-example contains only one
interaction with $x$ occurred in a computation,
and one interaction with $y$ occurred in the other computation.
This counter-example will be included in the model in (\ref{for6}),
only if $\mathsf{delta}_{T,m}$ is defined as above.

Let $h:\mathsf{varD} \mid \Delta \vdash M$ be an arbitrary term where identifiers from $\Delta$
may occur more than once in $M$.
Let $h:\mathsf{varD} \mid \Delta_1 \vdash M_1$ be derived without using the contraction for $\Delta_1$,
such that $h:\mathsf{varD} \mid \Delta \vdash M$ is obtained from it by applying
one or more times the contraction rule for identifiers from $\Delta$.
Then $\sbr{ h, h':\mathsf{varD} \mid \Delta \vdash M ; \mathsf{skip}^{\#} ; M'}^m$ is obtained
by first computing $\sbr{ h, h':\mathsf{varD} \mid \Delta_1 \vdash M_1 ; \mathsf{skip}^{\#} ; M_1'}^m$
as defined in (\ref{for6}),
and then by suitable tagging all moves associated with several occurrences of
the same identifier from $\Delta$ as described in the interpretation of contraction.
We have that:
 \[
 \begin{array}{l}
 \sbr{ h, h':\mathsf{varD}, \Delta  \vdash M ; \mathsf{skip}^{\#} ; M' }^m \subseteq
  \sbr{ h, h':\mathsf{varD}  \vdash
 M[N_1/x_1, \ldots, N_k/x_k] ; \! \mathsf{skip}^{\#} ; \! M'[N_1/x_1, \ldots, N_k/x_k]}
 \end{array}
  \]
for any $m \geq 0$ and arbitrary closed terms $\vdash N_1:T_1, \ldots, \vdash N_k:T_k$.

In the case that $\Delta$ contains only identifiers of base types $B$
which do not occur in any $\mathsf{while}$-subterm of $M$, then
in the above formula the subset relation becomes the equality for $m=0$.
If a free identifier occurs in a $\mathsf{while}$-subterm of $M$, then it can be called
arbitrary many times in $M$, and so we cannot reproduce its behaviour in $M'$.

\begin{theorem}
Let $h:\mathsf{varD} \! | \! \Delta \vdash M$ be a split (open) term, where
 $\Delta=x_1\!:\!T_1, \ldots, x_k\!:\!T_k$, and
 \begin{equation} \label{for7}
T = \begin{array}{l}
\lbrack \! \lbrack k: \mathsf{exp}D, \Delta \vdash \mathsf{new}_D \, h \aasg k \, \mathsf{in} \,  M ; \mathsf{skip}^{\#} ;
 \mathsf{new}_D \, h' \aasg k \, \mathsf{in} \, M' : \mathsf{com} \rbrack \! \rbrack^m
\end{array}
\end{equation}
\begin{itemize}
\item[(i)] Let $\Delta$ contains only identifiers of base types $B$,
which do not occur in any $\mathsf{while}$-subterm of $M$.
Any word of $T$ (where $m=0$) is of the form $w_1 \cdot \# \cdot w_2$ such that $|w_1|=|w_2|$ iff $M$ has no timing leaks.
\item[(ii)] Let $\Delta$ be an arbitrary context.
If there exists a word $w=w_1 \cdot \# \cdot w_2 \in T$ such that $|w_1| \neq |w_2|$, Then $M$ does have timing leaks.
\end{itemize}
\end{theorem}

Note that if a counter-example witnessing a timing leakage is found,
then it provides a specific context $\Delta$, i.e.\ a concrete definition
of identifiers from $\Delta$, for which the given open term have timing leaks.

\section{Detecting Timing-Aware Non-interference} \label{time-int}

The slot-game semantics model contains enough information to check
the non-interference property of terms along with timing leaks.
The method for verifying the non-interference property is analogous
to the one described in \cite{D13}, where we use the standard game
semantics model. As slot-game semantics can be considered as
the standard game semantics augmented with the information about
quantitative assessment of time usage, we can use it as underlying model
for detection of both non-interference property and timing leaks,
which we call \emph{timing-aware non-interference}.

In what follows, we show how to verify timing-aware non-interference
property for closed terms. In the case of open terms, the method can
be extended straightforwardly by following the same ideas for handling
open terms described in Section \ref{time}.

Let $l:\mathsf{varD}, h:\mathsf{varD'} \vdash M:\mathsf{com}$ be a term
where $l$ and $h$ represent low- and high-security global variables respectively.
We define $\Gamma_1=l:\mathsf{var}D, h:\mathsf{var}D'$, $\Gamma'_1=l':\mathsf{var}D, h':\mathsf{var}D'$,
and $M'$ is $\alpha$-equivalent to $M[l'/l,h'/h]$
 where all bound variables are suitable renamed.
 We say that $\Gamma_1 | - \vdash M:\mathsf{com}$ satisfies \emph{timing-aware non-interference} if
\[
\begin{array}{ll}
\forall s_1 \in St(\Gamma_1), s_2 \in St(\Gamma'_1). & s_1(l)=s_2(l') \land s_1(h) \neq s_2(h') \land \\
&  \Gamma_1 \vdash M ; M',\mathrm{s_1} \otimes \mathrm{s_2} \rightsquigarrow^{n_1} \mathsf{skip} ; M',\mathrm{s_1}' \otimes \mathrm{s_2} \rightsquigarrow^{n_2} \mathsf{skip ; skip},\mathrm{s_1}' \otimes \mathrm{s_2}'  \\
& \implies s_1'(l)=s_2'(l') \ \land \ n_1 = n_2
\end{array}
\]

Suppose that $\mathsf{abort}$ is a special free identifier of type $\mathsf{com}^{abort}$ in $\Gamma$.
We say that a term $\Gamma \vdash M$ is \emph{safe} iff
$\Gamma  \vdash M[\mathsf{skip}/\mathsf{abort}] \sqsub M[\mathsf{diverge}/\mathsf{abort}]$
\footnote{$\sqsub$ denotes observational approximation of terms (see \cite{AM2})};
otherwise we say that a term is \emph{unsafe}.
It has been shown in \cite{DGL.SAS} that
a term $\Gamma \vdash M$ is safe
iff $\sbr{\Gamma \vdash M}$ does not contain any play with moves from $\mathcal A_{\sbr{\mathsf{com}}}^{abort}$,
which we call unsafe plays.
For example, $\sbr{\mathsf{abort:com^{abort}} \vdash \mathsf{skip \, ; abort:com }}
= \sr \, \cdot \, \sr^{abort} \, \cdot \, \sd^{abort} \, \cdot \, \sd $, so this term is unsafe.

By using Theorem \ref{th.closed} from Section \ref{time} and the corresponding result
for closed terms from \cite{D13}, it is easy to show the following result.
\begin{equation} \label{for8}
\begin{array}{l}
L = \lbrack \! \lbrack k: \mathsf{exp}D, k':\mathsf{exp}D', \mathsf{abort}:\mathsf{com} \vdash  \mathsf{new}_D \, l \aasg k \, \mathsf{in} \, \mathsf{new}_{D'} \, h \aasg k' \, \mathsf{in} \\
\qquad \qquad \qquad \qquad \qquad \qquad \qquad \qquad \mathsf{new}_D \, l' \aasg !l \, \mathsf{in} \, \mathsf{new}_{D'} \, h' \aasg k' \, \mathsf{in} \\
\qquad \qquad \qquad \qquad \qquad \qquad \qquad \qquad \mathsf{skip}^{\#} ; M ; \mathsf{skip}^{\#} ; M'; \mathsf{skip}^{\#} ; \mathsf{if} \, (!l \neq !l') \, \mathsf{then} \, \mathsf{abort} : \mathsf{com} \rbrack \! \rbrack
\end{array}
\end{equation}
The regular expression $L$ contains no unsafe word (plays) and all its words
are of the form $w=w_1 \cdot \# \cdot w_2 \cdot \# \cdot w_3 \cdot \# \cdot w_4$ such that $|w_2|=|w_3|$
iff $M$ satisfies the timing-aware non-interference property.

Notice that the free identifier $k$ in (\ref{for8}) is used to initialize the
variables $l$ and $l'$ to any value from $D$ which is the same for both $l$ and $l'$,
while $k'$ is used to initialize $h$ and $h'$ to any values from $D'$. The last $\mathsf{if}$
command is used to check values of $l$ and $l'$ in the final state after evaluating
the term in (\ref{for8}). If their values are different, then $\mathsf{abort}$ is run.

\section{Application} \label{app}

We can also represent slot-game semantics model of IA$_2$
by using the CSP process algebra. This can be done by
extending the CSP representation of standard game semantics given in \cite{DL},
by attaching the costs corresponding to each translation rule.
In the same way, we have adapted the verification tool in \cite{DL}
to automatically convert an IA$_2$ term
into a CSP process \cite{Ros} that represents its slot-game semantics.
The CSP process outputted by our tool is defined by a script in machine readable CSP
which can be analyzed by the FDR tool. It represents a model checker for the CSP process algebra,
and in this way a range of properties of terms can be verified by calls to it.

In the input syntax of terms,
we use simple type annotations to indicate what finite sets of
integers will be used to model free identifiers and local
variables of type integer. An operation between values of types
$\mathsf{int}_{n_1}$ and $\mathsf{int}_{n_2}$ produces a value of
type $\mathsf{int}_{max\{n_1,n_2\}}$. The operation is performed
modulo $max\{n_1,n_2\}$.

In order to use this tool to check for timing leaks in terms,
we need to encode the required property as a CSP process
(i.e.\ regular-language). This can be done only if we know
the cost of the worst plays (paths) in the model of a given term.
We can calculate the worst-case cost of a term by generating
its model, and then by counting the number of tokens in
its plays. The property we want to check will be:
$\sum_{i=0}^{n} \circled{i} \cdot \# \cdot \circled{i}$, where $n$ denotes the
worst-case cost of a term.

To demonstrate practicality of this approach for automated verification,
we consider the following implementation of the linear-search algorithm.
\begin{center}
$\begin{array}{l}
 h : \mathsf{var int_2}, x[k] \, : \, \mathsf{var int_2} \vdash \\
 \qquad \mathsf{new}_{int_2} \,  a[k] \aasg 0 \, \mathsf{in} \\
 \qquad \mathsf{new}_{int_{k+1}} \, i \aasg 0 \, \mathsf{in} \\
 \qquad \mathsf{while} \, (i < k) \, \mathsf{do} \, \{ a[i] := !x[i] ; \ i := !i + 1 ;\} \\
 \qquad \mathsf{new}_{int_2} \,  y \aasg !h \, \mathsf{in} \\
 \qquad \mathsf{new}_{bool} \,  present:=ff \, \mathsf{in} \\
 \qquad \mathsf{while} \, (i<k \, \&\& \, \neg present) \, \mathsf{do} \, \{ \\
 \qquad \qquad  \mathsf{if} \, (compare(!a[i],!y)) \, \mathsf{then} \, present := tt; \\
 \qquad \qquad i := !i + 1; \,   \\
 \qquad \} \, : \mathsf{com}
\end{array}$
\end{center}
The meta variable $k>0$ represents the array size.
The term copies the input array $x$ into a local array $a$,
and the input value of $h$ into a local variable $y$.
The linear-search algorithm is then used to find whether the value stored in $y$ is in  the local array.
At the moment when the value is found in the array, the term terminates successfully.
Note that arrays are introduced in the model as syntactic sugar
by using existing term formers. So an array $x[k]$ is represented
as a set of $k$ distinct variables $x[0], \ldots, x[k-1]$ (see \cite{DL,GM} for details).

Suppose that we are only interested in measuring the efficiency of the
term relative to the number of $compare$ operations.
It is defined as follows $compare: \mathsf{expint_2 \to expint_2 \to expbool}$,
and its semantics compares for equality the values of two arguments with cost $\circled{\$}$:
\[
\begin{array}{l}
\sbr{compare: \mathsf{expint_2^1 \to expint_2^2 \to expbool}} =
 \sq \cdot \circled{\$} \cdot \sq^1 \cdot ( \sum_{m \neq n } m^1 \cdot \sq^2 \cdot n^2 \cdot ff ) + ( \sum_{m=n } m^1 \cdot \sq^2 \cdot n^2 \cdot tt )
\end{array}
\]
where $m,n \in \{ 0,1 \}$. We assume that the costs of all other
operations are relatively negligible (e.g.\ $k_{var}=k_{der}=\ldots=0$).

\begin{figure*}
\centerline{\scalebox{1.6}{\psfig{figure=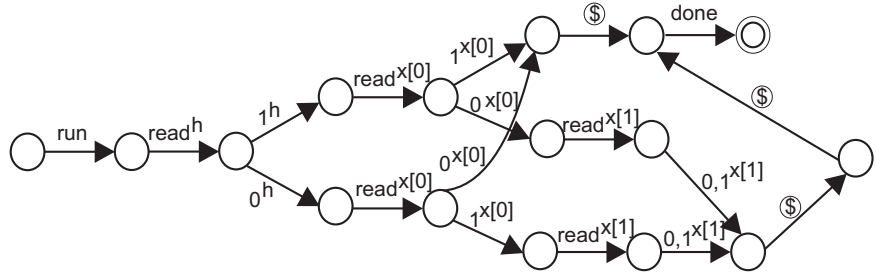}}}
\caption{Slot-game semantics for the linear search with $k$=2 } \label{linear.model}
\end{figure*}
We show the model for this term with $k=2$ in Fig.~\ref{linear.model}.
The worst-case cost of this term is equal to the array's size $k$,
which occurs when the search fails or the value of $h$ is
compared with all elements of the array.
We can perform a security analysis for this term by
considering the model extended as in (\ref{for7}), where $m=0$.
We obtain that this term has timing leaks, with a counter-example
corresponding to two computations, such that initial values of $h$
are different, and the search
succeeds in the one after only one iteration of $\mathsf{while}$ and fails in the other.
For example, this will happen when all values in the array $x$ are 0's, and the value of $h$ is 0 in the first computation
and 1 in the second one.

We can also automatically analyse in an analogous way
terms where the array size $k$ is much larger. Also the set of data
that can be stored into the global variable $h$ and array $x$
can be larger than $\{0,1\}$. In these cases we will obtain models
with much bigger number of states, but they still can be automatically
analysed by calls to the FDR tool.

\section{Conclusion}

In this paper we have described how game semantics can be used for verifying security
properties of open sequential programs, such as timing leaks and non-interference.
This approach can be extended to terms with infinite data types,
such as integers, by using some of the existing methods and tools
based on game semantics for verifying such terms.
Counter-example guided abstraction refinement procedure (ARP) \cite{DGL.SAS}
 and symbolic representation of game semantics model \cite{D12}
are two methods which can be used for this aim.
The technical apparatus introduced here applies not only to time as a resource
but to any other observable resource, such as power or
heating of the processor. They can all be modeled in the framework of
slot games and checked for information leaks.

We have focussed here on analysing the IA language, but we can easily extend this
approach to any other language for which game semantics exists.
Since fully abstract game semantics was also defined for probabilistic \cite{Dan},
concurrent \cite{Gh.ConCSP}, and
programs with exceptions \cite{AM2}, it will be interesting to
extend this approach to such programs.

\providecommand{\urlalt}[2]{\href{#1}{#2}}
\providecommand{\doi}[1]{doi:\urlalt{http://dx.doi.org/#1}{#1}}

\bibliographystyle{eptcs}

\end{document}